\documentclass[conference]{IEEEtran}
\IEEEoverridecommandlockouts
\usepackage{cite}
\usepackage{amsmath,amssymb,amsfonts, amsthm}
\usepackage{algorithmic}
\usepackage{graphicx}
\usepackage{textcomp}
\usepackage{xcolor}
\usepackage{bm}
\newcommand{\indep}{\mathop{\perp\!\!\!\!\perp}}
\def\BibTeX{{\rm B\kern-.05em{\sc i\kern-.025em b}\kern-.08em
    T\kern-.1667em\lower.7ex\hbox{E}\kern-.125emX}}
    
\newtheorem{thm}{Theorem}    
    
\begin{document}

\title{Heterogeneous Treatment Effect Estimation based on a Partially Linear Nonparametric Bayes Model
\thanks{This research is partially supported by the research grant from the Public Foundation of Chubu Science and Technology Center and No. 19K12128 of Grant-in-Aid for Scientific Research Category (C) and No. 18H03642 of Grant-in-Aid for Scientific Research Category (A), Japan Society for the Promotion of Science.}
}

\author{\IEEEauthorblockN{Shunsuke Horii}
\IEEEauthorblockA{Waseda University\\
1-6-1, Nishiwaseda, Shinjuku-ku, \\
Tokyo 169-8050, Japan\\
Email: s.horii@waseda.jp}
}

\maketitle

\begin{abstract}
Recently, conditional average treatment effect (CATE) estimation has been attracting much attention due to its importance in various fields such as statistics, social and biomedical sciences.
This study proposes a partially linear nonparametric Bayes model for the heterogeneous treatment effect estimation.
A partially linear model is a semiparametric model that consists of linear and nonparametric components in an additive form.
A nonparametric Bayes model that uses a Gaussian process to model the nonparametric component has already been studied.
However, this model cannot handle the heterogeneity of the treatment effect.
In our proposed model, not only the nonparametric component of the model but also the heterogeneous treatment effect of the treatment variable is modeled by a Gaussian process prior.
We derive the analytic form of the posterior distribution of the CATE and prove that the posterior has the consistency property.
That is, it concentrates around the true distribution.
We show the effectiveness of the proposed method through numerical experiments based on synthetic data.
\end{abstract}

\section{Introduction}
The problem of statistical causal inference, which is the estimation of the magnitude of the effect of changing the value of a treatment variable on the outcome variable, is one of the most important problems in data science.
Fisher's randomized controlled trial is a standard method for statistically estimating causal effects \cite{fisher1951design}.
However, randomized controlled trials are often difficult to implement due to cost and ethical considerations.
There is a growing demand for statistical causal inference in observational studies that do not involve experiments.

To estimate causal effects statistically in observational studies, it is necessary to make some assumptions about the data generation process and define causal effects mathematically.
In this study, we consider causal effect estimation based on the Neyman-Rubin potential outcomes model\cite{imbens2015causal}.
In the potential outcomes model, the distribution of the outcome variable is considered for each value of the treatment variable, and the causal effect is defined based on these distributions.
In this study, the treatment variable $T$ is assumed to take two values, $0$ and $1$.
The group of individuals whose treatment variable value is $0$ is called the control group, and the group of individuals whose treatment variable value is $1$ is called the treatment group.
Potential outcome variables are represented by the randome variables $Y^{(t)}, t=0,1$, where $t$ is the value of the treatment variable $T$.
In practice, only a value of either $0$ or $1$ for the treatment variable can be assigned to an individual, and we cannot observe the potential outcome of the unassigned treatment value.
Thus, further assumptions are needed to estimate the causal effects in the potential outcomes model.
A well-known assumption is referred to as \textit{strong ignorability} or \textit{non-confounding}, which assumes that there are covariates $X$ so that $T$ and $(Y^{(0}, Y^{(1)})$ are conditionally independent given $X$.

One of the causal effects treated in the potential outcomes model is the average treatment effect (ATE), which is defined as $\mathrm{E}[Y^{(1)}-Y^{(0)}]$.
There are various methods to estimate ATE, see \cite{imbens2015causal} and references therein.
On the other hand, the conditional average treatment effect (CATE), defined as $\mathrm{E}[Y^{(1)}-Y^{(0)}|X]$, is a causal effect that focuses on the heterogeneity of treatment effects caused by differences in the characteristics of individuals.
In recent years, research on methods for estimating CATE has been attracting much attention\cite{wager2018estimation,chernozhukov2018double, nie2021quasi, alaa2018bayesian}.

Under the strong ignorability, the estimation of ATE and CATE is reduced to the estimation of $\mathrm{E}[Y|T=t, X], (t=0,1)$.
Although we can estimate $\mathrm{E}[Y|T, X]$ by using nonparametric methods, the sample size required for estimation would become large, especially when $X$ is high-dimensional or $\mathrm{E}[Y|T, X]$ is not smooth.
On the other hand, when $p(Y|T, X)$ is modeled by a parametric model such as a linear regression model, the model may fail to fit the data due to too strong assumptions.
In this study, we focus on the partially linear model, which combines the linear regression model and the nonparametric regression model.
The partially linear model is more expressive than the linear regression model due to the inclusion of a nonparametric component.
It also can estimate parameters with higher accuracy than the full nonparametric model.
In this study, we take a Bayesian approach to estimate the causal effects in the partially linear model.
When the objective is to estimate the ATE or CATE, the nonparametric part of the model is considered a nuisance parameter.
The Bayesian approach allows us to marginalize the nuisance parameter.
In \cite{choi2015partially}, a partially linear model, in which the nonparametric component is modeled as a Gaussian process, is proposed.
Although this model has the advantage that the posterior distribution of ATE can be calculated analytically, it cannot handle the heterogeneity of causal effects because $\mathrm{E}[Y|T, X]$ is not dependent on $X$.
This study proposes a model in which the heterogeneous treatment effect is also modeled by a Gaussian process.
We show that the posterior distribution of CATE can be calculated analytically when the parameters of the Gaussian process are known.
We also prove the consistency of the posterior under some assumptions.

This paper is organized as follows.
In section 2, we describe the potential outcomes model and the definitions of ATE and CATE under the model.
We also explain some conditions, including strong ignorability, which enable us to estimate ATE and CATE.
Section 3 describes the partially linear model proposed in this study.
In section 4, we derive the posterior distribution of CATE under the proposed model.
Section 5 shows that the posterior has the consistency property under some assumptions.
Section 6 discusses some earlier studies that are closely related to this work.
In section 7, we verify that the proposed method can estimate CATE accurately by numerical experiments, and in section 8, we conclude and discuss future works.

\section{Potential Outcomes Model and Treatment Effects}

This study aims to investigate the causal effect of the binary treatment variable $T\in\left\{0, 1\right\}$ on the outcome variable $Y\in\mathbb{R}$.
We consider the Neyman-Rubin potential outcomes model to define the causal effect mathematically.
Let $Y^{(t)}, t=0, 1$ be the potential outcome variables where $Y^{(t)}$ is the random variable that represents the outcome variable when $T=t$.
As its name suggests, we assume that we cannot observe the potential outcomes $(Y^{(0)}, Y^{(1)})$ directly.
The observable outcome variable $Y$ is defined as $Y=TY^{(1)}+(1-T)Y^{(0)}$.
The average treatment effect (ATE) is defined as 
\begin{align}
    ATE=\mathrm{E}\left[Y^{(1)}-Y^{(0)}\right],
\end{align}
Assuming that $(Y^{(0)}_{i}, Y^{(1)}_{i}, T_{i}), i=1,\ldots,n$ is i.i.d. with the distribution $\mathbb{P}(Y^{(0)}, Y^{(1)}, T)$, consider the problem of estimating ATE from sample $(T_{i}, Y_{i}), i=1,\ldots, n$ with sample size $n$.
If $T$ and $(Y^{(0)}, Y^{(1)})$ are independent, then
\begin{align}
    \frac{1}{\left|\left\{i:T_{i}=1\right\}\right|}\sum_{i:T_{i}=1}Y_{i}, \label{simple_ATE_estimator}
\end{align}
is a consistent estimator of ATE under some mild conditions.
However, since $T$ and $(Y^{(0)}, Y^{(1)})$ are not necessarily independent, the estimator (\ref{simple_ATE_estimator}) is a biased estimator of ATE in general.

In addition to $T, Y$, we assume that there are some other covariates $X\in\mathcal{X}\subseteq \mathbb{R}^{d}$.
If $T$ and $(Y^{(0)}, Y^{(1)})$ are conditionally independent under $X$, i.e., $Y^{(0)}, Y^{(1)}\indep T|X$, then $T$ is said to satisfy the strongly ignorable condition \cite{imbens2015causal}.
Furthemore, we assume that $X$ and $T$ satisfy the overlap condition, that is, $0<\mathrm{Pr}(T=t|X=x)<1$ for all $t\in\left\{0,1\right\}, x\in\mathcal{X}$.
When the strongly ignorable condition and overlap condition are satisfied, it holds
\begin{align}
    &\mathrm{E}[Y^{(1)}-Y^{(0)}]=\mathrm{E}[Y^{(1)}]-\mathrm{E}[Y^{(0)}]\nonumber\\
    &=\mathrm{E}_{X}[\mathrm{E}[Y^{(1)}|X]]-\mathrm{E}_{X}[\mathrm{E}[Y^{(0)}|X]]\nonumber\\
    &=\mathrm{E}_{X}[\mathrm{E}[Y|X, T=1]]-\mathrm{E}_{X}\mathrm{E}[Y|X, T=0]],
\end{align}
where the third line comes from the conditional independence $Y^{(0)}, Y^{(1)}\indep T|X$.
Thus, the ATE can be estimated by estimating $\mathrm{E}_{X}[\mathrm{E}[Y|X, T=1]]$ and $\mathrm{E}_{X}[\mathrm{E}[Y|X, T=0]]$ from the sample.

As its name suggests, ATE is the average causal effect of $T$ on $Y$.
On the other hand, there are many problems where the causal effect of $T$ on $Y$ depends on the covariates $X$, and we are interested in the conditional average treatment effect (CATE), which is defined as 
\begin{align}
    CATE(X)=\mathrm{E}[Y^{(1)}-Y^{(0)}|X].
\end{align}
When the strongly ignorable condition is satisfied, it holds
\begin{align}
    \mathrm{E}[Y^{(1)}-Y^{(0)}|X]=\mathrm{E}[Y|X, T=1]-\mathrm{E}[Y|X, T=0].
\end{align}
Therefore, the estimation of CATE also reduces to the estimation of the conditional expectations $\mathrm{E}[Y|X, T=1]$ and $\mathrm{E}[Y|X, T=0]$.

\section{Partially Linear Model}
A possible strategy to estimate ATE and CATE is to first construct an estimator for $\mathrm{E}\left[Y|X, T\right]$.
For realizations $x, t$ of $X, T$, a simple estimator for $\mathrm{E}\left[Y|X=x, T=t\right]$ is given by
\begin{align}
    \frac{1}{|i:X_{i}=x, T=t|}\sum_{i:X_{i}=x, T_{i}=t}Y_{i}.
\end{align}
Although this estimator is consistent in some cases, it is difficult to use when $X$ contains continuous variables or is high-dimensional, or when few or no sample points with the same values exist.
Another way to construct the estimator is to assume a model among the variables $Y, X, T$.
A widely used model is the folloing linear regression model.
\begin{align}
    Y_{i}=\theta T_{i}+\beta^{\top} X_{i}+\varepsilon_{i}, \quad i=1,\ldots, n.
\end{align}
In this case, the estimation of ATE and CATE reduces to the estimation of $\theta$.
If the model can sufficiently express the true relationship among the variables, we can estimate $\theta$ with high accuracy, however, the model may be too simple in some cases.

In this paper, we focus on the partially linear model as a model that fixes the shortcoming of the linear regression model.
A simple partially linear model for $Y, X, T$ is given by
\begin{align}
    Y_{i}=\theta T_{i}+f(X_{i})+\varepsilon_{i},\quad i=1,\ldots,n, \label{PLM}
\end{align}
where $f(\cdot)$ is an unknown nonlinear function from $\mathcal{X}$ to $\mathbb{R}$.
In this study, we assume that the error term $\varepsilon_{i}, i=1,\ldots,n$ is i.i.d. with a normal distribution $\mathcal{N}(0, s_{\varepsilon}^{-1})$.
In this model, ATE and CATE are equal to $\theta$, and CATE is independent of $X$.
This means that the model cannot capture the heterogeneity of the causal effect of $T$ on $Y$.

We consider the following extension of the simple partially linear model.
\begin{align}
    Y_{i}=\theta(X_{i})T_{i}+f(X_{i})+\varepsilon_{i},\quad i=1,\ldots, n,\label{proposed_model}
\end{align}
where $\theta(\cdot)$ is an unknown nonlinear function from $\mathcal{X}$ to $\mathbb{R}$.
In this model, CATE is equal to $\theta(X)$, so the estimation of CATE reduces to the estimation of the function $\theta(\cdot)$.
Note that if $\theta(X)$ is a linear function of $X$, then $\theta(X)T$ represents the second-order interaction between $X$ and $T$, so (\ref{proposed_model}) can be thought of as a model that accounts for nonlinear interaction between $X$ and $T$.

As a model for the nonlinear functions $\theta(\cdot), f(\cdot)$, we use a Gaussian process prior.
We assume that 
\begin{align}
    \theta(\cdot)&\sim \mathcal{GP}(0, C(\cdot,\cdot; \omega_{\theta})),\\
    f(\cdot)&\sim \mathcal{GP}(0, C(\cdot, \cdot; \omega_{f})),
\end{align}
where $C(\cdot, \cdot; \omega)$ is the covariance function defined by the parameter $\omega$.
For example, we can use the following covariance function, known as the Gaussian kernel.
\begin{align}
    C(x_{1}, x_{2}; \beta )=\exp\left\{-\beta||x_{1}-x_{2}||^{2}\right\}.\label{gaussian_kernel}
\end{align}

\section{Derivation of Posterior Predictive Distribution}
Given a sample $\bm{t}=(t_{1},\ldots, t_{n}), \bm{X}=(x_{1},\ldots, x_{n}), \bm{y}=(y_{1},\ldots, y_{n})$, we consider the problem to estimate the value of CATE $(\theta(\tilde{x}_{1}),\ldots,\theta(\tilde{x}_{m}))$ for another sample $\tilde{\bm{X}}=(\tilde{x}_{1},\ldots, \tilde{x}_{m})$.
In this section, we assume that $\bm{t}, \bm{X}, \tilde{\bm{X}}$ are non-random variables and $\bm{y}$ is a random variable.
Let $\bm{\theta}=(\theta(x_{1}),\ldots, \theta(x_{n})),  \tilde{\bm{\theta}}=(\theta(\tilde{x}_{1}),\ldots,\theta(\tilde{x}_{m}))$, $\bm{f}=(f(x_{1}),\ldots,f(x_{n}))$.
Then $p(\bm{\theta}, \tilde{\bm{\theta}}, \bm{f})$ is a multivariate Gaussian distribution with mean $\bm{0}$ and covariance matrix
\begin{align}
    \left(
    \begin{array}{lll}
    \bm{\Phi}_{nn} & \bm{\Phi}_{nm} & \bm{O}\\
    \bm{\Phi}_{nm}^{T} & \bm{\Phi}_{mm} & \bm{O}\\
    \bm{O} & \bm{O} & \bm{\Psi}_{nn}
    \end{array}
    \right),
\end{align}
where $\bm{\Phi}_{nn}, \bm{\Phi}_{mn}, \bm{\Phi}_{mm}, \bm{\Psi}_{nn}$ are defined as 
\begin{align}
    \bm{\Phi}_{nn}&=\left(
    \begin{array}{ccc}
    C(x_{1},x_{1};\omega_{\theta}) & \cdots & C(x_{1}, x_{n};\omega_{\theta})\\
    \vdots & \ddots & \vdots\\
    C(x_{n}, x_{1};\omega_{\theta}) & \cdots & C(x_{n}, x_{n};\omega_{\theta})
    \end{array}
    \right),\\
    \bm{\Phi}_{nm}&=\left(
    \begin{array}{ccc}
    C(x_{1},\tilde{x}_{1};\omega_{\theta}) & \cdots & C(x_{1}, \tilde{x}_{m};\omega_{\theta})\\
    \vdots & \ddots & \vdots\\
    C(x_{n}, \tilde{x}_{1};\omega_{\theta}) & \cdots & C(x_{n}, \tilde{x}_{m};\omega_{\theta})
    \end{array}
    \right),\\
    \bm{\Phi}_{mm}&=\left(
    \begin{array}{ccc}
    C(\tilde{x}_{1},\tilde{x}_{1};\omega_{\theta}) & \cdots & C(\tilde{x}_{1}, \tilde{x}_{m};\omega_{\theta})\\
    \vdots & \ddots & \vdots\\
    C(\tilde{x}_{m}, \tilde{x}_{1};\omega_{\theta}) & \cdots & C(\tilde{x}_{m}, \tilde{x}_{m};\omega_{\theta})
    \end{array}
    \right),\\
    \bm{\Psi}_{nn}&=\left(
    \begin{array}{ccc}
    C(x_{1},x_{1};\omega_{f}) & \cdots & C(x_{1}, x_{n};\omega_{f})\\
    \vdots & \ddots & \vdots\\
    C(x_{n}, x_{1};\omega_{f}) & \cdots & C(x_{n}, x_{n};\omega_{f})
    \end{array}
    \right),
\end{align}
respectively\footnote{$\bm{O}, \bm{I}$ are the zero and identity matrices, respectively.
In this paper, we write $\bm{O}, \bm{I}$ for zero and identity matrices regardless of the matrix size}.
From the model (\ref{proposed_model}), $p(\bm{y}|\bm{\theta}, \bm{f})$ is a multivariate Gaussian distribution with mean $\bm{T}\bm{\theta}+\bm{f}$ and covariance matrix $s_{\varepsilon}^{-1}\bm{I}$, where we set $\bm{T}=\mbox{diag}(\bm{t})$.

First, we show that we can obtain the analytic form of $p(\tilde{\bm{\theta}}|\bm{y})$ given that $\omega_{\theta}, \omega_{f}, s_{\varepsilon}$ are known.
The joint distribution $(\bm{\theta}, \tilde{\bm{\theta}}, \bm{f}, \bm{y})$ is given by
\begin{align}
    p(\bm{\theta}, \tilde{\bm{\theta}}, \bm{f}, \bm{y})= p(\bm{y}|\bm{\theta},\bm{f})p(\bm{\theta}, \tilde{\bm{\theta}}, \bm{f}).\label{joint}
\end{align}
Expanding the exponential part of the equation (\ref{joint}) results in a quadratic form of $(\bm{\theta},\tilde{\bm{\theta}}, \bm{f}, \bm{y})$, which indicates that this joint distribution is a multivariate Gaussian.
If we denote the precision matrix of this joint distribution as $\bm{S}$, then
\begin{multline}
    \bm{S}=\left(
    \begin{array}{lll}
    \bm{\Phi}^{-1} & \bm{O} & \bm{O}\\
    \bm{O} & \bm{\Psi}_{nn}^{-1} & \bm{O}\\
    \bm{O} & \bm{O} & s_{\epsilon}\bm{I}
    \end{array}
    \right)+\\
    \left(
    \begin{array}{cccc}
    s_{\epsilon}\bm{T}^{2} & \bm{O} & s_{\epsilon}\bm{T} & -s_{\epsilon}\bm{T}\\
    \bm{O} & \bm{O} & \bm{O} & \bm{O}\\
    s_{\epsilon}\bm{T} & \bm{O} & s_{\epsilon}\bm{I} & -s_{\epsilon}\bm{I}\\
    -s_{\epsilon}\bm{T} & \bm{O} & -s_{\epsilon}\bm{I} & \bm{O}
    \end{array}
    \right),\label{S}
\end{multline}
where we set
\begin{align}
    \bm{\Phi} = \left(
    \begin{array}{cc}
    \bm{\Phi}_{nn} & \bm{\Phi}_{nm}\\
    \bm{\Phi}_{nm}^{T} & \bm{\Phi}_{mm} 
    \end{array}
    \right).
\end{align}
The covariance matrix is given by $\bm{\Sigma}=\bm{S}^{-1}$.

Let $\Theta=(\bm{\theta}, \tilde{\bm{\theta}}, \bm{f})$ and $\bm{\Sigma}$ be
\begin{align}
    \bm{\Sigma}=\left(
    \begin{array}{cc}
    \bm{\Sigma}_{\Theta\Theta} & \bm{\Sigma}_{\Theta\bm{y}}\\
    \bm{\Sigma}_{\bm{y}\Theta} & \bm{\Sigma}_{\bm{y}\bm{y}}
    \end{array}
    \right).
\end{align}
Then the posterior distribution $p(\bm{\theta}, \tilde{\bm{\theta}}, \bm{f}|\bm{y})$ is also a multivariate Gaussian from the property of multivariate Gaussian distribution and whose mean and covariance matrix is given by
\begin{align}
    \bm{\mu}_{\Theta|\bm{y}}&=\bm{\Sigma}_{\Theta\bm{y}}\bm{\Sigma}_{\bm{yy}}^{-1}\bm{y},\\
\bm{\Sigma}_{\Theta|\bm{y}}&=\bm{\Sigma}_{\Theta\Theta}-\bm{\Sigma}_{\Theta\bm{y}}\bm{\Sigma}_{\bm{yy}}^{-1}\bm{\Sigma}_{\bm{y}\Theta},
\end{align}
see \cite{bishop:2006:PRML}, for example.
Furthermore, let $\bm{M}=\bm{\Sigma}_{\Theta\bm{y}}\bm{\Sigma}_{\bm{yy}}^{-1}$ and 
\begin{align}
    \bm{M}&=\left(
    \begin{array}{c}
    \bm{M}_{\bm{\theta}}\\
    \bm{M}_{\tilde{\bm{\theta}}}\\
    \bm{M}_{\bm{y}}
    \end{array}
    \right),\\
    \bm{\Sigma}_{\Theta|\bm{y}}&=\left(
    \begin{array}{ccc}
    \bm{\Sigma}_{\bm{\theta\theta}|\bm{y}} & \bm{\Sigma}_{\bm{\theta}\tilde{\bm{\theta}}|\bm{y}} & \bm{\Sigma}_{\bm{\theta}\bm{f}|\bm{y}}\\
    \bm{\Sigma}_{\tilde{\bm{\theta}}\bm{\theta}|\bm{y}} & \bm{\Sigma}_{\tilde{\bm{\theta}}\tilde{\bm{\theta}}|\bm{y}} & \bm{\Sigma}_{\tilde{\bm{\theta}}\bm{f}|\bm{y}}\\
    \bm{\Sigma}_{\bm{f}\bm{\theta}|\bm{y}} & \bm{\Sigma}_{\bm{f}\tilde{\bm{\theta}}|\bm{y}} & \bm{\Sigma}_{\bm{f}\bm{f}|\bm{y}},
    \end{array}
    \right).
\end{align}
Then, $p(\tilde{\bm{\theta}}|\bm{y})$ is a multivariate Gaussian distribution with mean $\bm{M}_{\tilde{\bm{\theta}}}\bm{y}$ and covariance matrix $\bm{\Sigma}_{\tilde{\bm{\theta}}\tilde{\bm{\theta}}|\bm{y}}$.
The mean vector is the Bayes optimal estimator in terms of squared error loss function.

When $\omega_{\theta}, \omega_{f}, s_{\varepsilon}$ are unknown, we would assume priors for these variables.
However, the analytic form of the posterior is no longer available in this case.
We can use Markov Chain Monte Carlo (MCMC) methods to obtain an approximation of the posterior.
We can also estimate the value of $\omega_{\theta}, \omega_{f}, s_{\varepsilon}$ by maximizing the marginal likelihood.
Prior distributions are assumed to derive the theoretical property of the posterior (section 5).
However, we take the latter approach for numerical experiments because of its computational cost (section 7).

\section{Theoretical Analysis}
This section proves that the posterior of $\theta$ and $f$ has consistency.
We use some results proved in \cite{ghosal1999posterior}, which proves the posterior consitency of the Gaussian process prior for binary regression problems.
For the posterior to have consistency, several assumptions are needed.
First, we assume that the covariance functions have the form of 
\begin{align}
    C(x_{1}, x_{2}; \tau_{\theta}, \lambda_{\theta})&=\tau_{\theta}^{-1}k_{0}(\lambda_{\theta}x_{1}, \lambda_{\theta}x_{2}),\\
    C(x_{1}, x_{2}; \tau_{f}, \lambda_{f})&=\tau_{f}^{-1}k_{0}(\lambda_{f}x_{1}, \lambda_{f}x_{2}),
\end{align}
where $k_{0}(\cdot,\cdot)$ is a nonsingular covariance kernel and hyper-parameters $\tau_{\theta},\tau_{f}, \lambda_{\theta}, \lambda_{f}$ take positive values.
Let the priors on $\tau_{\theta},\tau_{f}, \lambda_{\theta}$ and $\lambda_{f}$ be $\Pi_{\tau_{\theta}}, \Pi_{\tau_{f}}, \Pi_{\lambda_{\theta}}$ and $\Pi_{\lambda_{f}}$.
To state other assumptions, we introduce some notations.
We assume that the covariates $X$ are independently distributed according to a distribution $\mathbb{P}_{X}$ and it has a density $p(x)$.
We also assume that the treatment variable $T$ is distributed according to a distribution $\mathbb{P}_{X|T}$ and it has a probability function $p(t|x)$ given $x$.
Let $p_{\theta,f}(x, t, y)=p(y|x,t,\theta,f)p(t|x)p(x)$, which is the joint density of $(x,t,y)$ induced by $\theta, f$.
We denote $\Pi$ for the prior distribution of $p_{\theta, f}$ when the parameters of $\theta, f$ are distributed according to $\Pi_{\tau_{\theta}}, \Pi_{\tau_{f}}, \Pi_{\lambda_{\theta}}, \Pi_{\lambda_{f}}$.
We define a subset of joint densities
\begin{align}
    \mathcal{P}_{n, \alpha}&=\left\{p_{\theta, f}: \theta, f\in\mathcal{G}_{n,\alpha}\right\},\\
    \mathcal{G}_{n,\alpha}&=\left\{g:||D^{w}g||_{\infty}<M_{n}, w\le \alpha \right\},
\end{align}
where $D^{w}g$ denotes $(\partial^{w}/\partial^{w_{1}}\ldots\partial^{w_{d}})g(x_{1},\ldots,x_{d}), w=\sum w_{i}$, $\alpha$ is a positive integer and $M_{n}$ is a sequences of real numbers.
Let $\lambda_{\theta,n}, \lambda_{f,n}, \tau_{\theta,n}, \tau_{f,n}$ be sequences that satisfy $\Pi_{\tau_{\theta}}(\tau_{\theta}<\tau_{\theta,n})=e^{-cn}, \Pi_{\lambda_{\theta}}(\lambda_{\theta}>\lambda_{\theta,n})=e^{-cn}, \Pi_{\tau_{f}}(\tau_{f}<\tau_{f,n})=e^{-cn}, \Pi_{\lambda_{f}}(\lambda_{f}>\lambda_{f,n})=e^{-cn}$, for some constant $c$.

We make the following assumptions.
\begin{description}
    \item[(P)] For every fixed $x\in\mathcal{X}$, the covariance kernel $k_{0}(x,\cdot)$ has continuous partial derivatives up to order $2\alpha+2$, where $\alpha$ is a positive integer which satisfies a certain condition described later. The prior $\Pi_{\lambda_{\theta}}, \Pi_{\lambda_{f}}$ for $\lambda_{\theta}, \lambda_{f}$ are fully supported on $(0, \infty)$.
    \item[(C)] The covariate space $\mathcal{X}$ is a bounded subset of $\mathbb{R}^{d}$.
    \item[(T)] The true functions $\theta_{0}$ and $f_{0}$ belong to the reproducing kernel Hilbert space (RKHS) of $k_{0}$.
    \item[(G)] For every $b_{1}>0$ and $b_{2}>0$, there exist sequences $M_{n}, \lambda_{\theta,n}, \lambda_{f,n}, \tau_{\theta,n}$ and $\tau_{f,n}$ that satisfy
    \begin{align}
        M_{n}^{2}\tau_{\theta,n}\lambda_{\theta,n}^{-2}\ge b_{1}n,\quad
        M_{n}^{d\alpha}\le b_{2}n,\\
        M_{n}^{2}\tau_{f,n}\lambda_{f,n}^{-2}\ge b_{1}n,\quad
        M_{n}^{d\alpha}\le b_{2}n.
    \end{align}
\end{description}
Similar assumptions as in assumptions (P), (C), (T) and (G) appear in \cite{ghosal1999posterior}.
To prove the posterior consistency, we need to show that $\Pi(\mathcal{P}_{n,\alpha}^{c})$ is exponentially small and bound the metric entropy of $\mathcal{P}_{n,\alpha}$.
Assumptions (P), (C), (T) and (G) are necessary for this purpose.

Let $D_{n}=\left(X_{i}, T_{i}, Y_{i}\right)_{i=1,\ldots,n}$ and $\mathbb{P}_{0}^{n}$ denote the true distribution of $D_{n}$.
We define the $L_{1}$ metric between $p_{\theta,f}$ and $p_{\theta_{0}, f_{0}}$ as
\begin{multline}
    ||p_{\theta, f}-p_{\theta_{0},f_{0}}||_{L_{1}}=\\
    \sum_{t}\int \int |p_{\theta, f}(x,t,y)-p_{\theta_{0}, f_{0}}(x,t,y)|{\rm d}y{\rm d}x.
\end{multline}

Then the following theorem holds.
\begin{thm}\label{consitency_theorem}
Suppose that assumptions (P), (C), (T) and (G) hold.
Then for any $\epsilon>0$,
\begin{align}
    \Pi\left((\theta, f):||p_{\theta,f}-p_{\theta_{0},f_{0}}||_{L_{1}}>\epsilon\ |\ D_{n}\right)\to 0
\end{align}
with $\mathbb{P}_{0}^{n}$-probability $1$.
\end{thm}
\begin{proof}
See the Appendix.
\end{proof}

\section{Related Works}
The study of a partially linear model dates back to the work of Engle et al.\cite{engle1986semiparametric}.
While there are many studies on Bayesian approaches to partially linear models, the most relevant to our research is the work by Choi et al.\cite{choi2015partially}, where a Gaussian process prior is assigned to the nonlinear function $f$ in the model (\ref{PLM}).
As mentioned in the introduction, in this model, it holds $E[Y^{(1)}-Y^{(0)}|X]=E[Y^{(1)}-Y^{(0)}]$.
Thus, it cannot handle the heterogeneity of causal effects.

Another closely related work is the method known as Double/Debiased Machine Learning (DML) by Chernozhukov et al.\cite{chernozhukov2018double}.
In DML, the following model is assumed for $Y, X, T$.
\begin{align}
    Y_{i}&=\theta(X_{i})T_{i}+f(X_{i})+\varepsilon_{i},\quad E[\varepsilon|X_{i}]=0,\\
    T_{i}&=g(X_{i})+\eta_{i},\quad E[\eta|X_{i}]=0. \label{p_t_x}
\end{align}
The difference from the model proposed in our study (\ref{proposed_model}) is that $p(T_{i}|X_{i})$ is also modeled by (\ref{p_t_x}) and no parametric assumption is made on the distribution of $\varepsilon_{i}, \eta_{i}$.
There are many studies on the extensions of DML depending on the function class of $\theta$, and one of particular relevance to our work is the work of Nie et al.\cite{nie2021quasi}.
They consider the case where $\theta$ is an element of an RKHS.

Alaa and van der Schaar proposed to directly model the potential outcomes with nonparametric regression\cite{alaa2018bayesian}, namely,
\begin{align}
    Y^{(t)}_{i}=f_{t}(x_{i})+\varepsilon_{i}, \quad t\in\left\{0,1\right\},
\end{align}
where $f_{t}(\cdot), t=0,1$ are unknown nonlinear functions that follow some prior distribution $\Pi(f_{0},f_{1})$.
A difference between their study and ours is that their study assumes two different regression functions, while our study models interaction between $T$ and $X$ that is linear in $T$ and nonlinear in $X$.

\section{Experiments}
To verify the effectiveness of the proposed method, we compare it with some conventional estimators.
We compare the performance of the proposed method with that of DML using synthetic data.
In the first experiment, we generate data under the following conditions.
\begin{itemize}
    \item The dimension of $\mathcal{X}$ is $d=2$
    \item The functions $\theta$ and $f$ are generated from a Gaussian process prior
    \begin{itemize}
        \item The Gaussian kernel (\ref{gaussian_kernel}) is used for the covariance function
    \end{itemize}
    \item The precision of the noise is $s_{\varepsilon}=1.0$
\end{itemize}

DML need to specify the function class of $\theta$ and the estimation methods of $E[Y|X], E[T|X]$.
We use the following combinations.
\begin{itemize}
    \item Linear function + Linear regression / Logistic regression
    \item Linear function + GP regression / GP classification
    \item RHKS + Linear regression / Logistic regression
    \item RHKS + GP regression / GP classification
\end{itemize}

The proposed method has to determine the hyperparameters of the covariance functions and the precision of the noise.
In the experiments, we estimate them by maximizing the marginal likelihood.

We compare the methods with the mean squared error of CATE $\frac{1}{m}\sum_{i=1}^{m}(\theta(\tilde{x})_{i}-\hat{\theta}_{i})^{2}.$
Figure \ref{fig:error_curve_GP_GP} shows the mean squared error curves of CATE as the functions of the sample size.
We know that the proposed method is optimal in this setting, so it is not surprising that the posterior mean yields the smallest mean squared error.
Another almost trivial result is that since true $\theta$ is a nonlinear function, kernel DML, which assumes that $\theta$ is an element of RKHS, outperforms the linear DML, which assumes linear function for $\theta$.
We can also see that in DML, the estimation methods of the conditional expectations $E[Y|X], E[T|X]$ do not have a significant impact on the performance.
We can also confirm that the mean squared error of linear DML does not approach $0$ as the sample size increases.
This may be because the true $\theta$ is nonlinear and not included in the assumed class of linear functions.

In the second experiment, the functions $\theta$ and $f$ are randomly generated linear functions.
The coefficients of the linear functions are generated from the standard Gaussian distribution, and the other conditions are the same as in the first experiment.
Figure \ref{fig:error_curve_Linear_Linear} shows the experiment result.
As a natural result, linear DML performs better than the proposed method and kernel DML since the true $\theta$ is linear.
Compared to the first experiment, the choice of the estimation methods of the conditional expectations $E[Y|X], E[T|X]$ has a significant impact on the performance.
We think kernel DML with GP regression and GP classification considers the same function classes of $\theta$ and $f$ as the proposed method.
We can see that the proposed method outperforms kernel DML in this case.
Unlike the case of the first experiment, the mean squared errors of al estimators approach $0$ as the sample size increases, although the convergence rates vary.
This may be because the true $\theta$ is linear, and the true function is included in the class of functions assumed by linear DML, which assumes a class of linear functions, and other methods that assume a class of RKHS.

\begin{figure}
    \centering
    \includegraphics[keepaspectratio=True,width=\linewidth]{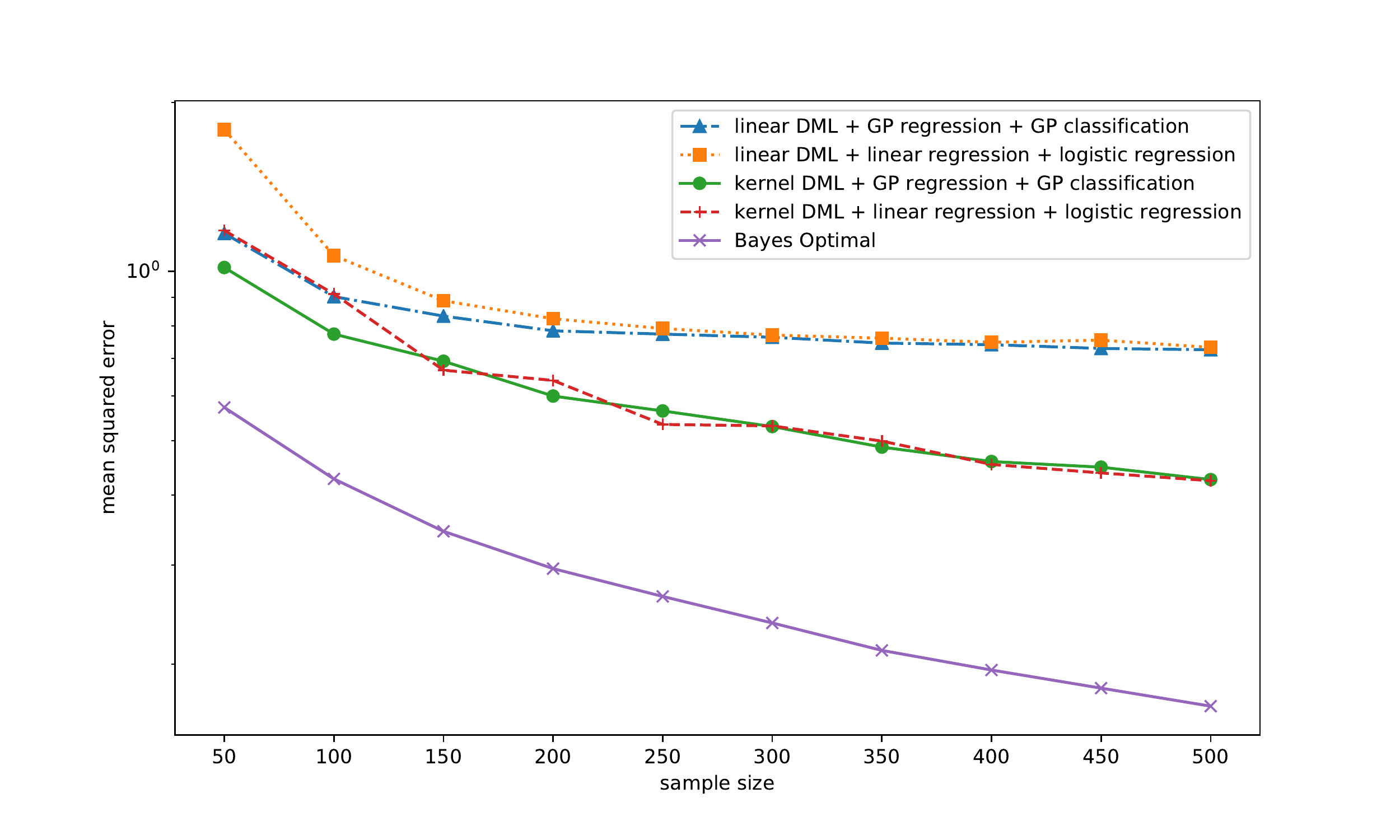}
    \caption{Curves of the mean squared error of CATE as the function of sample size. Functions $\theta$ and $f$ are generated from a Gaussian process prior.}
    \label{fig:error_curve_GP_GP}
\end{figure}

\begin{figure}
    \centering
    \includegraphics[keepaspectratio=True,width=\linewidth]{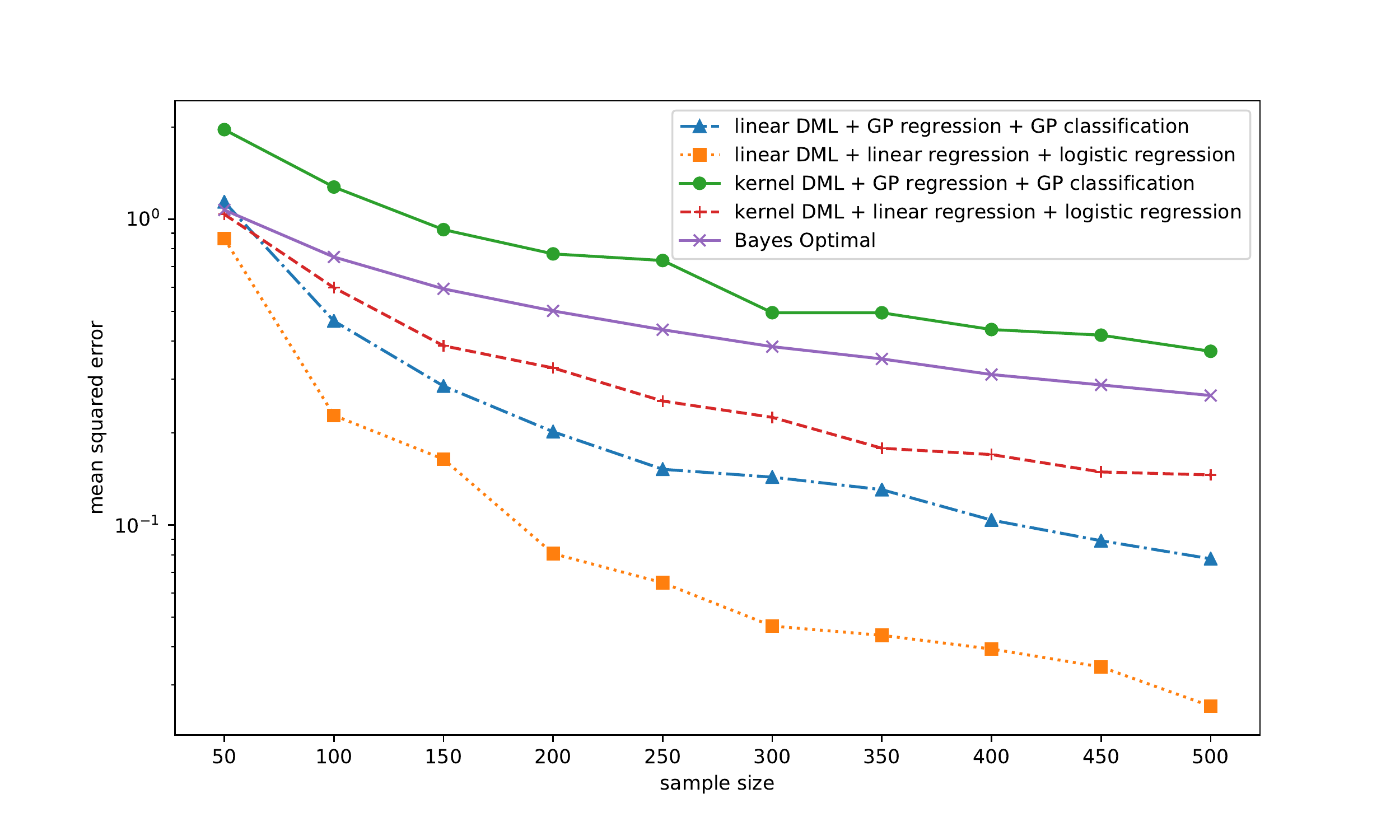}
    \caption{Curves of the mean squared error of CATE as the function of sample size. Functions $\theta$ and $f$ are randomly generated linear funcitons.}
    \label{fig:error_curve_Linear_Linear}
\end{figure}

\section{Conclusion}
In this study, we proposed a partially linear model in which the conditional average treatment effect (CATE) is modeled by a Gaussian process prior and derived the posterior distribution of the CATE.
The CATE is modeled as a nonlinear function that follows a Gaussian process prior, which allows for flexible modeling, while the CATE can be estimated with high accuracy.

We prove that the posterior has consistency under some conditions.
However, the convergence rate and the relationship with the class of nonlinear functions are not yet derived.
In \cite{alaa2018bayesian}, the relationship between the minimax information rate and the assumed class of nonlinear functions is derived for the CATE estimation problem in a model different from our study.
Conducting a similar analysis for our model is future work.

The proposed method requires $O(n^{3})$ of computation for a sample size $n$, and the computation becomes difficult when the sample size becomes large.
We will derive an efficient approximate calculation method in the future.

\newpage

\bibliographystyle{IEEEtran}
\bibliography{ref}

\appendices

\section{Proof of Theorem \ref{consitency_theorem}}
From Theorem 2 in \cite{ghosal1999posterior}, if suffices to verify the following conditions hold:
\begin{itemize}
    \item $\Pi\left((\theta, f):\mbox{KL}(p_{\theta_{0},f_{0}}||p_{\theta,f})<\epsilon\right)>0$ for every $\epsilon>0$, where $\mbox{KL}$ is the Kullback-Leibler (KL) divergence
    \item There exists $\beta>0$ such that $\log N(\epsilon, \mathcal{P}_{n,\alpha}, ||\cdot||_{1})<n\beta$, where $N(\epsilon, \mathcal{P}_{n,\alpha}, ||\cdot||_{1})$ is the covering number (and its logarithm is the metric entropy)
    \item $\Pi(\mathcal{P}^{c}_{n,\alpha})$ is exponentially small
\end{itemize}

We first show that $\Pi\left((\theta, f):\mbox{KL}(p_{\theta_{0},f_{0}}||p_{\theta,f})<\epsilon\right)>0$ for every $\epsilon>0$.
From the definition of KL divergence,
\begin{align}
    &\mathrm{KL}(p_{\theta_{0},f_{0}}||p_{\theta,f})\nonumber\\&=\mathrm{E}_{p_{\theta_{0},f_{0}}}\left[\log \frac{p_{\theta_{0},f_{0}}(X, T, Y)}{p_{\theta,f}(X, T, Y)}\right]\\
    &=\mathrm{E}_{p_{\theta_{0},f_{0}}}\left[\log \frac{p(Y|X,T, \theta_{0},f_{0})}{p(Y|X,T,\theta,f)}\right]\\
    &=\int \int \log\frac{p(y|x,T=1,\theta_{0},f_{0})}{p(y|x,T=1,\theta,f)}p(T=1|x)p(x){\rm d}y{\rm d}x\nonumber\\
    &+\int \int \log\frac{p(y|x,T=0,\theta_{0},f_{0})}{p(y|x,T=0,\theta,f)}p(T=0|x)p(x){\rm d}y{\rm d}x\\
    &\le \int \int \log\frac{p(y|x,T=1,\theta_{0},f_{0})}{p(y|x,T=1,\theta,f)}p(x){\rm d}y{\rm d}x\nonumber\\
    &+\int \int \log\frac{p(y|x,T=0,\theta_{0},f_{0})}{p(y|x,T=0,\theta,f)}p(x){\rm d}y{\rm d}x,\label{sum_KL}
\end{align}
where the last inequality follows from $p(T=t|x)\le 1, t=0,1$.
Let $h_{0}(x)=\theta_{0}(x)+f_{0}(x)$ and $h(x)=\theta(x)+f(x)$.
Then 
\begin{align}
\int yp(y|x,T=0,\theta_{0},f_{0}){\rm d}y=f_{0}(x),\\
\int yp(y|x,T=1,\theta_{0},f_{0}){\rm d}y=h_{0}(x).
\end{align}
Substituting these into (\ref{sum_KL}), some algebra leads to
\begin{align}
  &\mathrm{KL}(p_{\theta_{0},f_{0}}||p_{\theta,f})\nonumber\\
  &\le \frac{s_{\varepsilon}}{2}\left(\int(h_{0}(x)-h(x))^{2}p(x){\rm d}x\right.\nonumber\\
  &\qquad\qquad+\left.\int(f_{0}(x)-f(x))^{2}p(x){\rm d}x\right)\\
  &\le \frac{s_{\varepsilon}}{2}\left(||h_{0}-h||_{\infty}^{2}+||f_{0}-f||_{\infty}^{2}\right)\\
  &\le \frac{s_{\varepsilon}}{2}\left(||\theta_{0}-\theta||_{\infty}^{2}+2||f_{0}-f||_{\infty}^{2}\right)\label{KL_bound}.
\end{align}
From Theorem 4 in \cite{ghosal1999posterior}, it holds $\Pi(\theta: ||\theta_{0}-\theta||_{\infty}<\epsilon)>0$ and $\Pi(f:||f_{0}-f||_{\infty}<\epsilon)>0$, thus $\Pi((\theta,f):\mbox{KL}(p_{\theta_{0},f_{0}}||p_{\theta,f})<\epsilon)>0$.

Next, we bound the covering number $N(\epsilon, \mathcal{P}_{n, \alpha}, ||\cdot||_{1})$.
To simplify the description, we write $N(\epsilon, \mathcal{G}_{n}, ||\cdot||_{\infty})$ as $N_{\epsilon}$.
From the definition of the covering number, we can construct $\theta_{1},\ldots,\theta_{N_{\epsilon}}$, $f_{1},\ldots,f_{N_{\epsilon}}$ that satisfy the following condition:
\begin{itemize}
    \item There exists $i,j\in\left\{1,\ldots,N_{\epsilon}\right\}$ that satisfy $||\theta-\theta_{i}||_{\infty}<\epsilon, ||f-f_{i}||_{\infty}<\epsilon$ for all $\theta, f\in\mathcal{G}_{n}$
\end{itemize}
We choose $\theta,f\in\mathcal{G}_{n}$ arbitrarily and let $\theta^{*}, f^{*}$ be the functions that satisfy the above condition.
Then it holds
\begin{align}
    ||p_{\theta,f}-p_{\theta^{*},f^{*}}||_{1}&\le\sqrt{2\mbox{KL}(p_{\theta,f}||p_{\theta^{*},f^{*}})}\\
    &\le \sqrt{s_{\varepsilon}\left(||\theta-\theta^{*}||_{\infty}^{2}+2||f-f^{*}||_{\infty}^{2}\right)}\\
    &\le \sqrt{3s_{\varepsilon}}\epsilon,
\end{align}
where the first line follows from Pinsker's inequality \cite{Cover2006}, the second line follows from (\ref{KL_bound}), and the last line follows from the definition of $\theta^{*},f^{*}$.
The above inequality implies $N(\sqrt{3s_{\epsilon}}\epsilon, \mathcal{P}_{M_{n},\alpha},||\cdot||_{1})\le N_{\epsilon}^{2}$.
From the proof of Theorem 1 in \cite{ghosal1999posterior}, it holds $\log N(\epsilon, \mathcal{G}_{n}, ||\cdot||_{\infty})\le K\epsilon^{-d/\alpha}b^{d/\alpha}n$, so
\begin{align}
    \log N(\epsilon, \mathcal{P}_{M_{n},\alpha}, ||\cdot||_{1})\le 2K(\epsilon/\sqrt{3s_{\epsilon}})^{-d/\alpha}b^{d/\alpha}n.
\end{align}
By letting $b<(\beta/(2K))^{\alpha/d}(\epsilon/\sqrt{3s_{\epsilon}})$, it holds 
\begin{align}
\log N(\epsilon, \mathcal{P}_{n}, \left||\cdot|\right|_{1})<n\beta.    
\end{align}

Finally, it is easy to verify that $\Pi(\mathcal{P}_{M_{n},\alpha}^{c})$ is exponentially small from Lemma 1 in \cite{ghosal1999posterior}.

\end{document}